\documentclass[11pt]{article}

\newif\ifnotes
\notestrue

\title{List Decoding Bounds for Binary Codes with Noiseless Feedback}

\author{Meghal Gupta\thanks{Email: \texttt{meghal@berkeley.edu}. Supported by a UC Berkeley Chancellor's Fellowship.}\\UC Berkeley 
\and Rachel Yun Zhang\thanks{Email: \texttt{rachelyz@mit.edu}. Supported by NSF Graduate Research Fellowship 2141064. Supported in part by DARPA under Agreement No. HR00112020023 and by an NSF grant CNS-2154149.}\\MIT}
\date{\today}

\usepackage{hyperref}
\usepackage[margin=1in]{geometry}
\usepackage{mathptmx}
\usepackage{amssymb}
\usepackage{amsmath}
\usepackage{amsthm}
\usepackage{amsfonts}
\usepackage{bm}
\usepackage{enumitem}
\usepackage{color}
\usepackage{comment}
\usepackage[capitalize]{cleveref}
\usepackage[dvipsnames]{xcolor}
\usepackage{float}
\usepackage[T1]{fontenc}
\usepackage{todonotes}
\usepackage{asymptote}
\usepackage{mdframed}
\usepackage[most]{tcolorbox}
\usepackage{hyperref}
\usepackage{enumitem}
\usepackage{framed}
\usepackage{mdframed}
\usepackage{scrextend}
\usepackage{multirow}
\usepackage{ifthen}
\usepackage{bbm}
\usepackage{frcursive}
\usepackage{thm-restate}
\usepackage{booktabs}
\usepackage{array}
\usepackage{graphicx}

\usepackage[T1]{fontenc}

\definecolor{denim}{rgb}{0.08, 0.38, 0.74}
\definecolor{classicrose}{rgb}{0.98, 0.8, 0.91}
\definecolor{darkpastelblue}{rgb}{0.47, 0.62, 0.8}
\definecolor{dogwoodrose}{rgb}{0.84, 0.09, 0.41}

\usepackage{hyperref}
\hypersetup{
    colorlinks=true,
    linkcolor=dogwoodrose,
    filecolor=dogwoodrose,
    citecolor=denim,
    urlcolor=denim,
}
\usepackage[hyperpageref]{backref}

\newtheorem{theorem}{Theorem}[section]
\newtheorem{lemma}[theorem]{Lemma}
\newtheorem*{lemma*}{Lemma}
\newtheorem{claim}[theorem]{Claim}

\newtheorem{definition}[theorem]{Definition}

\theoremstyle{definition}

\Crefname{theorem}{Theorem}{Theorems}
\Crefname{claim}{Claim}{Claims}
\Crefname{lemma}{Lemma}{Lemmas}
\Crefname{proposition}{Proposition}{Propositions}
\Crefname{corollary}{Corollary}{Corollaries}
\Crefname{definition}{Definition}{Definitions}

\newcommand{\pos}{\mathsf{pos}}
\newcommand{\posc}{\mathsf{posc}}
\newcommand{\eps}{\varepsilon}
\newcommand{\SW}{$\mathsf{SW}$}
\newcommand{\quadr}{\mathsf{q}}

\renewcommand{\le}{\leqslant}
\renewcommand{\leq}{\leqslant}
\renewcommand{\ge}{\geqslant}
\renewcommand{\geq}{\geqslant}

\newcommand{\bbN}{\mathbb{N}}

\makeatletter
\newcommand{\customlabel}[2]{%
   \protected@write \@auxout {}{\string \newlabel {#1}{{#2}{\thepage}{#2}{#1}{}} }%
   \hypertarget{#1}{#2}
}
\makeatother

\newcommand{\protocol}[3]{
    \stepcounter{figure}
    \vspace{0.15cm}
    { \small
    \begin{tcolorbox}[breakable, enhanced, colback=classicrose!20]
    \begin{center}
    {\bf \underline{Protocol~\customlabel{prot:#2}{\thefigure}: #1}}
    \end{center}
    
    #3
    \end{tcolorbox}
    }
}

\newcounter{datacounter}

\newcounter{coingame}

\newcommand\numberthis{\addtocounter{equation}{1}\tag{\theequation}}

\newcounter{casenum}
\newenvironment{caseof}{\setcounter{casenum}{0}}{\vskip.5\baselineskip}
\newcommand{\case}[2]{
    \refstepcounter{casenum}
    \ifthenelse{\equal{\value{casenum}}{0}}{
    \vskip.5\baselineskip\par\noindent
    }{}
    \noindent {\it Case \arabic{casenum}:} {\it #1}
    \vskip0.1\baselineskip
    \begin{addmargin}[1.5em]{1em}
    #2
    \end{addmargin}
}

\newcounter{subcasenum}

\newcommand{\subcase}[2]{
    \refstepcounter{subcasenum}
    \vskip.5\baselineskip\par\noindent 
    {\it Subcase \arabic{casenum}.\arabic{subcasenum}:} {\it #1} \vskip0.1\baselineskip
    \begin{addmargin}[1.5em]{1em}
    #2
    \end{addmargin}
}
  
\newcounter{casenumb}

\newcounter{subcasenumb}

\begin{document}

\sloppy
\maketitle
\begin{abstract}
In an error-correcting code, a sender encodes a message $x \in \{ 0, 1 \}^k$ such that it is still decodable by a receiver on the other end of a noisy channel. In the setting of \emph{error-correcting codes with feedback}, after sending each bit, the sender learns what was received at the other end and can tailor future messages accordingly. 


While the unique decoding radius of feedback codes has long been known to be $\frac13$, the list decoding capabilities of feedback codes is not well understood. In this paper, we provide the first nontrivial bounds on the list decoding radius of feedback codes for lists of size $\ell$. For $\ell = 2$, we fully determine the $2$-list decoding radius to be $\frac37$. For larger values of $\ell$, we show an upper bound of $\frac12 - \frac{1}{2^{\ell + 2} - 2}$, and show that the same techniques for the $\ell = 2$ case cannot match this upper bound in general. 

\end{abstract}
\thispagestyle{empty}
\newpage

\enlargethispage{1cm}
\tableofcontents
\pagenumbering{roman}
\newpage
\pagenumbering{arabic}

\parskip=0.5ex
\section{Introduction}
Consider the problem of error-correcting codes with feedback~\cite{Berlekamp64}. Alice wishes to communicate a message $x \in \{0,1\}^k$ to Bob by transmitting a sequence of $n$ bits over a noisy channel. After sending each bit, she receives \emph{feedback} through a noiseless channel, telling her what bit Bob actually received. She can use this feedback to adaptively decide her next bit of communication, potentially allowing for improvement over classical error-correcting codes. As in the classical setting, the goal is to communicate in a way such that even if some fraction $r$ of Alice's bits are adversarially corrupted, Bob can still correctly determine $x$.

For classical error-correcting codes \emph{without} feedback, the Plotkin bound~\cite{Plotkin60} states that Bob cannot uniquely decode $x$ if more than $\frac{1}{4}$ of the communication is adversarially corrupted. In other words, the unique decoding radius for binary error-correcting codes is $\frac{1}{4}$.\footnote{In this introduction, when we state a resilience threshold/radius $r$, we mean that for any $\epsilon > 0$ there exists a protocol resilient to $r - \epsilon$ errors, with rate $O_\epsilon(1)$. For sake of simplicity, we've omitted all $\epsilon$'s from this introduction.} The corresponding question for codes with feedback was first posed by Berlekamp in his 1964 thesis, "Block Coding with Noiseless Feedback": What is the maximum fraction $r$ of errors from which a feedback-based error-correcting code can still uniquely be decoded? Remarkably, he proved that in the feedback setting, the unique decoding radius increases to $\frac{1}{3}$, surpassing the Plotkin bound.


While Berlekamp's result tells us the unique decoding radius of feedback codes, little is known about the list-decoding case. In list-decoding, first introduced by Elias~\cite{Elias57} and Wozencraft~\cite{Wozencraft58}, the goal is for Bob to output a small \emph{list} of messages including Alice's. 
The concept of list decoding for non-feedback codes has proven to be quite illuminating for understanding the Hamming space and for various other applications (see e.g. \cite{Guruswami01} for further discussion).

In the classical setting, list decoding is useful because it can be done at a larger error radius than unique decoding. 
In fact, the optimal list decoding radius is known for (non-feedback) codes. For decoding into a list of size $\ell$, the optimal radius is known to be $\frac12 - \binom{2s}{s} \cdot 2^{-2s-1}$ where $s:=\left\lfloor \frac\ell2 \right\rfloor$, which is asymptotically $\frac12-\Theta(\ell^{-1/2})$~\cite{Blinovsky86}. This is achieved by positive rate random codes. 

In this paper, we investigate the analogous question for feedback codes, first posed in~\cite{Shayevitz09}. In a feedback code, what is the maximum fraction $r$ of adversarial errors such that Bob can always learn a list of size $\ell$ containing Alice's message? We denote such a code as a $(\ell, r)$-list decodable feedback code. In our work, we provide the first non-trivial unconditional lower and upper bounds for this question.

\subsection{Results And Further Discussion}

Our first result gives an upper bound on the error resilience (decoding radius) of an $\ell$-list decodable feedback code.

\begin{restatable}{theorem}{upperbound}
\label{thm:upper-bound}
    For any positive integer $\ell$, there is no $\left( \ell,\frac12-\frac{1}{2^{\ell+2}-2} \right)$-list feedback code.
\end{restatable}

Prior to this work, the only known bound on the list decoding radius of feedback codes, aside from the trivial upper bound of $\frac{1}{2}$, was given in~\cite{Shayevitz09}. That result showed upper bounds for a specific subclass of feedback codes that satisfy a ``generalized translation'' property. It turns out this generalized translation property is not very comprehensive: we will in fact demonstrate a code that breaks the upper bound proven in~\cite{Shayevitz09}. In contrast, our upper bound is unconditional and applies to all feedback codes. 

In the case of $\ell=1$, Theorem~\ref{thm:upper-bound} provides an upper bound of $\frac{1}{3}$ on the unique decoding radius, consistent with the result in~\cite{Berlekamp64}. The key question, however, is whether this bound is tight in general. Although we do not have an answer for arbitrary $\ell$, we are able to show it is true in the specific case of $\ell=2$. Specifically, we show that for any $\eps > 0$, there exists a $\left( 2,\frac{3}{7}-\eps \right)$-list feedback code, matching the upper bound established by Theorem~\ref{thm:upper-bound}.

\begin{restatable}{theorem}{lowerbound}
\label{thm:sw-3.7}
    For any $\eps>0$, there is a $\left( 2, \frac37-\eps \right)$-list feedback code where $n = O_\eps(k)$.
\end{restatable}

The particular feedback code we used to achieve $\left( 2, \frac37 - \epsilon \right)$-list decodability is the same code found in~\cite{SpencerW92}, which Spencer and Winkler showed achieves unique decoding up to radius $\frac13 - \epsilon$. The code, which we'll denote \SW, is defined in Section~\ref{sec:sw-3.7}. Compared to the $(\frac13 - \epsilon)$-unique decodability analysis in~\cite{SpencerW92}, however, our analysis is significantly more involved, requiring us to characterize the likelihood of four inputs rather than one possible input at once.

A reasonable conjecture at this point is that the \SW~protocol will always give a $\left( \ell,\frac12-\frac{1}{2^{\ell+2}-2} \right)$-list decodable code. This would show that Theorem~\ref{thm:upper-bound} is optimal. 

Unfortunately, this turns out not to be true. We show that for $\ell = 3$, the error resilience of the Spencer-Winkler protocol does not match the upper bound given in Theorem~\ref{thm:upper-bound}. There is an adversarial corruption strategy on \SW~that requires only $\frac{31}{67} \approx 0.4627$ fraction of corruptions to confuse Bob between four different inputs, whereas the bound given by Theorem~\ref{thm:upper-bound} is $\frac7{15} \approx 0.4667$. 

\begin{restatable}{theorem}{counterexample}
\label{thm:counterexample}
    For any $\eps>0$, the feedback code \SW~(defined in Section~\ref{sec:sw-3.7}) is not $\left( 3, \frac{31}{67} +\eps \right)$-list decodable for any choice of encoding length $n = n(k)$.
\end{restatable}

Although this result states that the \SW~feedback code is not $\left( 3, \frac{7}{15} \right)$-list decodable, it remains unknown whether a different feedback code achieves this bound or whether the upper bound can be improved. 

Even the asymptotic behavior for this question remains unknown. Our upper bound suggests that the feedback list decoding radius could approach $\frac12$ exponentially fast with respect to the list size $\ell$. That is, the radius for $\ell$-list decoding is given by $\frac12-2^{-O(\ell)}$. We conjecture that this exponential bound is in fact optimal for the feedback case and propose proving or disproving this as an interesting direction for future research.

In the case of standard ECC's (without feedback), as mentioned earlier, the optimal list decoding radius is $\frac12 - \binom{2s}{s} \cdot 2^{-2s-1}$ for $where s=\left\lfloor \frac\ell2 \right\rfloor$~\cite{Blinovsky86}. That is, for any $\eps>0$, there exist $\left( \ell,\frac12 - \binom{2s}{s} \cdot 2^{-2s-1} - \eps \right)$-list decodable codes. Asymptotically, this bound says that the list decoding radius is approximately $\frac12-\Theta(\ell^{-1/2})$, which approaches $\frac12$ much slower than the conjectured $\frac12-2^{-O(\ell)}$ radius for the feedback case. Even short of proving the exponential bound, demonstrating a separation between the asymptotics for the feedback and non-feedback cases would be an intriguing result.
\subsection{Related Works}

\subsubsection{Error Correcting Codes with Feedback}

We already mentioned the seminal work of~\cite{Berlekamp64} which introduced the concept of error-correcting codes with feedback. Berlekamp studied adversarial corruption strategies and proved that $\frac13$ was the optimal error resilience for a feedback ECC in the face of adversarial errors. Feedback ECC's have also been studied for the binary symmetric channel, where every bit is flipped independently with some fixed probability~\cite{Berlekamp68,Zigangirov76}. 

Particularly relevant is the work of~\cite{Shayevitz09}, who initiated the study of list decoding feedback codes and gave some preliminary bounds on the list decoding capacities and radii of codes with noiseless feedback. Using generalizations of Berlekamp's methods, he provides upper bounds on list-decoding capacities for strategies that satisfy what he call a ``generalized translation'' property. By contrast, our upper bounds make no assumptions about the types of protocols allowed and are fully general.

Many other works on error-correcting codes with feedback present themselves instead under the lens of \emph{adaptive search with lies}. The two setups are equivalent, and our results imply results in the field of adaptive search with lies. We discuss the relationship below.

\paragraph{Adaptive Search with Lies} 

An equivalent way of looking at error-correcting codes with unlimited feedback is as follows: 
After every message from Alice, Bob asks a yes or no question about Alice's input, and Alice's next bit is a response to that question. Then, the adversary can corrupt Alice's responses, but only up to some fraction $\alpha$ of the time. Bob is ``searching'' for Alice's input in a range $[N]$ by asking these yes/no questions, and receiving an incorrect answer up to an $\alpha$ fraction of the time. This game is known as the Rényi-Ulam game, independently introduced by Rényi in \cite{Renyi61} and Ulam \cite{Ulam91}. 

Two excellent surveys on the research in this area are presented in \cite{Deppe07,Pelc02}. There are a number of works that aspire to improve upon \cite{Berlekamp64}. For example, \cite{Muthukrishnan94} presents a strategy with optimal rate (asking the least possible number of questions) while still achieving the best error-correction threshold. An alternative construction meeting the same optimal error resilience was given in~\cite{SpencerW92}. We mention that some works focus on a fixed (constant) number of lies rather than resilience to a fraction of errors, and ask for the smallest number of questions necessary. \cite{Pelc87} shows the optimal bound for one lie, the original Rényi-Ulam game. Followups find the optimal strategy for two and three lies as well \cite{Guzicki90,Deppe00}.

\paragraph{Limited Feedback/Questions.}

The works of~\cite{NegroPR95,AhlswedeCDV09} study the minimal number of rounds, or \emph{batches}, of questions Bob needs to ask when conducting adaptive search in the Ulam-Renyi model. The work of \cite{Pelc02} poses the open question of how many batches of feedback are necessary before the optimal strategy becomes essentially the same as the fully adaptive setting. 

In a more code-centric view, this model corresponds to a feedback code with \emph{limited feedback}, where Alice receives feedback only a few times throughout the protocol instead of after every bit she sends. 
\cite{HaeuplerKV15} studied the case where the feedback Alice receives consists of a small fraction of Alice's rounds. \cite{GuptaGZ23} took this direction even further, showing that a constant number of rounds of feedback consisting of only $\log(n)$ bits is enough to achieve the same unique decoding radius as for the unlimited feedback case. \cite{RuzomberkaJLP23} showed that with $O(1)$ bits of feedback, one can achieve the zero-error capacity of the error channel with full noiseless feedback.

\subsubsection{List Decoding}

The notion of list decoding was introduced by Elias~\cite{Elias57} and Wozencraft~\cite{Wozencraft58}. The goal is for Bob to output a small \emph{list} of messages including Alice's. Although they introduced list-decoding in the setting of random error models, the main focus since then has been for adversarial error models. See~\cite{Sudan00,Guruswami07, Guruswami09} for surveys on the topic.

One main advantage is that list decoding enables correcting up to a factor two more worst-case errors compared to unique decoding. As we stated earlier, the optimal tradeoff between the list size $\ell$ and error tolerance (in the classical non-feedback setting) is known~\cite{Blinovsky86}. This work shows that there are positive rate codes achieving this maximum possible error tolerance. An analogous result for $q$-ary codes was shown in~\cite{Resch24}.

Another advantage of list-decoding is that it allows for better rate codes. One can ask what the optimal tradeoff between the error tolerance $r$ and the rate $R$ of the code is, when the list size is required to small or constant, but not necessarily achieving the optimal bounds above. In the large alphabet setting, a simple random coding argument shows that $r>1-R-o(1)$ is achievable. The work of ~\cite{Guruswami08} provides an explicit positive rate efficiently encodable/decodable code for this task using folded Reed Solomon codes. This research has since been expanded upon, including in~\cite{Guruswami11, Brakensiek23, Guo23, Alrabiah24}.

Works such as~\cite{Guruswami14,Guruswami20} also discuss the tradeoff between all three parameters: list size, error tolerance, and rate.

\section{Definitions}
\label{sec:prelims}

Before stating our results, let us formally define a list decodable feedback ECC.

\begin{definition} [$(\ell, r)$-List Decodable Feedback Code]
    A $(\ell, r)$-list feedback code is a family of protocols $\{\pi_k\}_{k\in \bbN}$ from Alice to Bob defined as follows for each $k \in \bbN$.
    \begin{itemize}
        \item Alice begins with an input $x\in \{0,1\}^k$ that she intends to communicate to Bob.
        \item In each of $n = n(k) = |\pi_k|$ rounds, Alice sends a single bit that may be flipped (by an adversary). At the end of the round, Alice learns whether the bit was flipped.
        \item At the end of the protocol, Bob must output a list $L$ of size $\ell$ that must contain $x$, as long as at most $rn$ bits were flipped throughout the protocol.
    \end{itemize}
\end{definition}


\paragraph{The Coin Game Reformulation.}
Throughout the remainder of the paper, we'll argue our results about list decodable feedback ECC's in the language of a \emph{coin game} played between two players Bob and Eve. To this end, we reformulate the task of communicating over a noisy channel with feedback as this aforementioned coin game. A version of this coin game reformulation was first given by~\cite{SpencerW92} but is quite similar to the adaptive search with lies formulation of~\cite{Ulam91} and~\cite{Renyi61}. 

\begin{definition}[$(\ell, r; K, n)$-Coin Game]
    In the \emph{coin game} parametrized by $K, n, \ell \in \bbN$ and $r \in [0, 1]$, two players Bob and Eve are playing a game on a number line. At the beginning of the game, all $K$ coins are at position $0$. In each of $n$ rounds, Bob partitions the coins into two sets, and Eve chooses a set for which Bob moves all the coins up by $1$. 

    Eve wins if, at the end of the $n$ rounds, there are $\ge \ell + 1$ coins all with positions $\le rn$, and Bob wins otherwise. 
\end{definition}




The connection to feedback codes is as follows. The $K=2^k$ coins correspond to all possible values of Alice's input. For each of the $n$ rounds of the feedback code, Bob partitions the coins into two sets based on the two possibilities $0$ and $1$ for Alice's next message. Based on the bit that he receives, he increments the position of all coins in the opposite set; that is, the coins for which Alice would have sent the opposite bit. In this manner, at any point in time, each coin's position is the number of corruptions that must have occurred so far if the given coin were actually Alice's input.


Let's say we want a scheme that is $(\ell, r)$-list decodable. Bob wins if at the end of the $n$ rounds, there are at most $\ell$ coins left with positions $< rn$ (we say all other coins have ``fallen off the board''). Eve's goal, on the other hand, is to pick a set that Bob hears each round so that at the end of the game there are $> \ell$ coins left on the board. If she's able to pick \emph{any} sequence of sets such that $> \ell$ coins are left on the board at the end of the game, then the resulting code is necessarily not $(\ell, r)$-list decodable, since any of the $> \ell$ coins left on the board could've been Alice's input.

\paragraph{Notation.} Let us define the notation that we will use for a coin game.
\begin{itemize}
\item 
    For a coin $x$, $\pos_t(x)$ denotes the position of coin $x$ after $t$ rounds.
\item 
    $c_t(i)$ denotes the coin with the $i$'th smallest position at $t$ rounds.
\item 
    For $i \in [K]$, $\posc_t(i) = \pos_t(c(i))$ denotes the position of the coin with the $i$'th smallest position after $t$ rounds.
    
\end{itemize}

\section{Upper Bound}
\label{sec:upper}
In this section, we show an upper bound (impossibility result) that no $\ell$-list feedback code can achieve an error resilience greater than $\frac12-\frac{1}{2^{\ell+2}-2}=\frac{2^\ell-1}{2^{\ell+1}-1}$. Let us state the main theorem of this section.

\upperbound*

In the formulation of the coin game, we wish to establish that in an $n$-round protocol with $K$ coins, the adversary can always maintain $\ell+1$ coins below position $\frac{2^\ell-1}{2^{\ell+1}-1}\cdot n$ for any $K$ sufficiently large compared to $\ell$. From this, we formally conclude Theorem~\ref{thm:upper-bound} in Section~\ref{sec:conclude-upper}.

We begin by establishing the statement for $\ell=1$; that is, the adversary can always maintain $2$ coins at or below position $\frac13\cdot n$. This statement was known as early as \cite{Berlekamp64} when the concept of feedback codes was introduced, but we recreate a proof here. Our proof for general $\ell$ is inductive and therefore requires establishing the $\ell=1$ case first.

\subsection{Base Case: 1-List Feedback Codes}\label{sec:base-case}

We want to show that the adversary can always maintain the second coin at or below position $\frac13 \cdot n$. However, in order to use our statement as an appropriate base case for the inductive hypothesis in Section~\ref{sec:inductive-step}, we actually show the following stronger statement about partially completed coin games. 

Consider a partially completed coin game where $m$ of the $n$ rounds have been played already. Then, the adversary can limit $\posc_n(2)$ in terms of these variables, specifically by the expression $\max\left\{ \left\lceil \frac13\cdot (n-m+\posc_m(1)+\posc_m(2)+\posc_m(3))\right\rceil, \pos_m(2) \right\}$. In particular, the case where $m=0$ corresponds to the situation where the entire coin game is remaining, and indeed the expression says that $\posc_n(2)$ is at most $\frac13\cdot n$.

\begin{theorem} \label{thm:base-case}
    For any $\eps>0$, the following holds for sufficiently large $n$. Consider an $n$-round coin game on $K=3$ coins in which $m$ rounds have already been completed. Then, the adversary can ensure at the end of the $n$ rounds that 
    \[
        \posc_n(2)\leq \max\left\{ \left\lceil \frac13\cdot (n-m+\posc_m(1)+\posc_m(2)+\posc_m(3))\right\rceil, \posc_m(2) \right\}.
    \]
\end{theorem}

Before proving this, let us informally describe the strategy Eve uses to achieve this bound. This strategy is essentially identical to what has appeared in previous literature such as~\cite{Gelles-survey}. For most of the protocol, she simply chooses the smaller (fewer elements) of the two sets given by Bob. That way, she need only move up one coin, and so on average, a coin move up by $\frac13$, which is what gives the $\frac13\cdot (n-m)$ term in the stated bound. However, the issue is that the coins need not move evenly. For example, the first coin may not move at all, while the last two coins each move by $\frac12 \cdot (n-m)$. Then, the position of the second coin increases too much, but in compensation, we have some leeway with the first coin. To use this, once the first coin is sufficiently low in position that it can never surpass the second coin, Eve switches strategies to always choosing the set that \emph{does not} contain the second coin. Now, although the average coin can move up more per round, the second coin no longer moves which is what we wanted, and we do not risk the first coin catching up to the second by doing this.

\begin{proof}
    Let $\theta=\max\left\{ \left\lceil\frac13\cdot (n-m+\posc_m(1)+\posc_m(2)+\posc_m(3))\right\rceil, \posc_m(2) \right\}$. In each round $t$, Bob chooses the set $S_t$ and its complement $\overline{S_t}$ (both subsets of the 3 coins in the game), and the adversary Eve needs to choose one set for which to move up the coins by $1$ each. Eve uses the following strategy.

    For a round $t$, while $\posc_{t-1}(2)< \theta$, Eve always chooses the smaller of $S_t$ and $\overline{S_t}$, that is the one with fewer elements. (Note that the inequality in this condition is strict, so, for example, if $\posc_m(2) \geq \left\lceil \frac13\cdot (n-m+\posc_m(1)+\posc_m(2)+\posc_m(3)) \right\rceil$, it is possible to skip this step entirely.) 

    Once $\posc_t(2)$ reaches $\theta$, then Eve instead always selects the set not containing the coin in the second position. 
    
    Let us say this transition between the two strategies occurs at round $T+1$, so Eve spends rounds $m+1$ through $T$ inclusive doing the first strategy, and the rest of the protocol doing the second strategy. (If this round $T$ never occurs, then at the end of the protocol, $\posc_n(2)<\theta$ as we desired.) Our goal will be to show that in the remaining $n-T$ rounds, the first coin will not pass the second coin, even if it is selected to increase every time. Then there will be two coins left at the end of the protocol, both with position $\le \theta$.

    Between rounds $m+1$ and $T$, in every step, the total positions of the coins are increasing by at most $1$. Also, the first time the inequality is violated, we must have $\posc_T(2) = \theta \geq \left\lceil\frac13\cdot (n-m+\posc_m(1)+\posc_m(2)+\posc_m(3))\right\rceil$. Therefore, at the end of round $T$, it holds that 
    \begin{align*}
        &\posc_T(1)+\posc_T(2)+\posc_T(3) \leq \posc_m(1)+\posc_m(2)+\posc_m(3)+ T-m \\
        \implies& \posc_T(1)+2\posc_T(2) \leq \posc_m(1)+\posc_m(2)+\posc_m(3)+ T-m \\
        \implies& \posc_T(1)-\posc_T(2) + 3\cdot \left\lceil\frac13\cdot (n-m+\posc_m(1)+\posc_m(2)+\posc_m(3))\right\rceil \\
        &~\leq \posc_m(1)+\posc_m(2)+\posc_m(3)+ T-m \\
        \implies& \posc_T(1) - \posc_T(2) \\
        &~\leq \posc_m(1) - \posc_m(2)+\posc_m(3)+ T-m - (n-m+\posc_m(1)+\posc_m(2)+\posc_m(3)) \\ 
        &~\leq T-n \\
        \implies & \posc_T(2)-\posc_T(1)\geq n-T
    \end{align*}

    In other words, in the remaining $n-T$ rounds, the first coin will not be able to switch positions with the second or third coin. As such, $\posc(2)$ will never increase since it is never selected, and the first coin never surpasses it.

    This concludes the proof, as it means that $\posc_n(2)$ stays at $\theta$ (or never reached $\theta$).
\end{proof}

\subsection{Inductive Step: $\ell$-List Feedback Codes} \label{sec:inductive-step}

In this section, we establish the inductive step. We inductively show an analog of Theorem~\ref{thm:base-case} for larger list sizes.

\begin{theorem} \label{thm:inductive-step}
    For any positive integer $\ell$, the following holds for sufficiently large $n$. Consider an $n$-round coin game on $K=2^{\ell+1}-1$ coins in which $m$ rounds have already been completed. For round $m \in [n]$, let 
    \[
        \theta_\ell:= \left\lceil \frac{1}{2^{\ell+1}-1}\cdot \left( (2^\ell-1)\cdot (n-m)+ \posc_m(1)+\ldots+\posc_m(K) \right) \right\rceil.
    \]
    If $\posc_m(1)\geq \theta_\ell+m-n$ then Eve can ensure at the end of the $n$ rounds that $\posc_n(\ell+1) \leq \theta_\ell$.
\end{theorem}

Again, before proving this formally, let us describe Eve's strategy at a high level. Her strategy to achieve this bound will be recursive. At the start of the protocol, she chooses the smaller of Bob's two sets at every step, so that the average coin moves by $\frac{2^\ell-1}{2^{\ell+1}-1}$, giving the $\frac{2^\ell-1}{2^{\ell+1}-1}\cdot (n-m)$ term in the bound. As with the $\ell=1$ case, this doesn't give the desired bound for the final position of the $(\ell+1)$'th coin because it's possible the first few coins did not move at all, while the $(\ell+1)$'th coin moved by more. However, it can be shown that in compensation, the position of the first coin must be very low in this case. Once the first coin is at a low enough position that even if it is selected in every remaining round, it can no longer exceed $\theta_\ell$ by the end of the protocol, Eve switches strategies. Now, she needs to ensure that of the remaining coins, $\ell$ of them remain below position $\theta_\ell$ (since we have already found one coin that will definitely remain below $\theta_\ell$). This is where her strategy becomes recursive: she simulates a new coin game on the leftmost $2^{\ell}-1$ coins of the $2^{\ell+1}-2$ remaining coins and plays on these coins according to her (recursively described) strategy for the case of $\ell-1$.

\begin{proof}
    First we check that the statement holds for the base case $\ell = 1$. This follows from Theorem~\ref{thm:base-case}. In this case, we have that
    \begin{align*}
        \theta_1 :=&~ \left\lceil \frac13\cdot (n-m+\posc_m(1)+\posc_m(2)+\posc_m(3)) \right\rceil \\
        \geq&~ \frac13\cdot (n-m+\posc_m(1)+\posc_m(2) + \posc_m(3)) \\
        \geq&~ \frac13\cdot (\theta_1 +2\posc_m(2)).
    \end{align*}

    In the last step, we used the condition in the Theorem~\ref{thm:inductive-step} that $\posc_m(1)\geq \theta_\ell+m-n$. Rearranging, this implies that $\theta_1 \geq \posc_m(2)$. Theorem~\ref{thm:base-case} states that $\posc_n(2)\leq \max\{ \theta_1, \posc_m(2)\}=\theta_1$, since $\theta_1 \geq \posc_m(2)$, which is what we wanted to prove.
    
    For the inductive step, let us assume the statement is true for $\ell-1$ and prove it for $\ell$. Throughout the proof, let $\Sigma_m := \posc_m(1)+\ldots+\posc_m(K)$.
    
    In each round $t$, Bob's two sets are $S_t$ and its complement $\overline{S_t}$, and the adversary Eve needs to choose one set for which to move up the coins by $1$ each. Eve uses the following strategy.

    For a round $t$, while $\posc_{t-1}(1) > \theta_\ell+(t-1)-n$, Eve always chooses the one of $S_t$ and $\overline{S_t}$ with fewer elements. Note that the first time it is violated, $\posc_t(1) = \theta_\ell+t-n$, even if it is already violated immediately at round $t=m+1$ by the assumption. Denote the round it is first violated as $T$ if it is ever violated.

    At this point (after round $T$) or at the start of the protocol if the condition began by being violated, it is impossible for the coin $c_T(1)$ to ever hit the threshold $\theta_\ell$, since its position can only increase by $1$ for each of the remaining $n-T$ rounds. Thus Eve can ignore the first coin, and play a coin game on only the remaining $2^{\ell+1}-1$ coins, with the goal of keeping $\posc_n(\ell)$ below $\theta_\ell$. She also opts to exclude the second half of these coins $c_T(2^\ell+1),\ldots,c_T(2^{\ell+1}-1)$ from her new coin game as well, and simply assumes that they will end up above the threshold. In all, the players are playing a new a coin game on $2^\ell-1$ coins which correspond to $c_T(2),\ldots,c_T(2^\ell)$, with Eve having the goal of keeping $\posc_n(\ell)$ below $\theta_\ell$. She executes the strategy for this recursively, keeping $\posc_n(\ell)$ bounded by what Theorem~\ref{thm:inductive-step} says for $\ell-1$.

    Now, let us analyze this strategy.

    First, we address the case where round $T$ never occurs. In this case, at the end of the protocol, it holds that $\posc_n(1) > \theta_\ell-n+n=\theta_\ell$. Since Eve is picking the smaller of $S_t$ and $\overline{S_t}$ each time, it holds that
        \begin{align*}
        &\posc_n(1)+\ldots+\posc_n(2^{\ell+1}-1) \leq (2^\ell-1)\cdot (n-m) + \Sigma_m \\
        \implies& \posc_n(1) \leq \frac{1}{2^{\ell+1}-1} \cdot \left( (2^\ell-1)\cdot(n-m) + \Sigma_m \right) \leq \theta_\ell,
    \end{align*}
    which is a contradiction,
    where the last inequality follows because the middle term is at most $\theta_\ell$ by definition.

    Next, we address the case where round $T$ occurs at some time $m\leq T \leq n$. After round $T$, it holds that $\posc_T(1) = \theta_\ell+T-n$ exactly, and so 
    \begin{align*}
        &\posc_T(1)+\ldots+\posc_T(2^{\ell+1}-1) \leq (2^\ell-1)\cdot (T-m) + \Sigma_m \\
        \implies& \posc_T(2)+\ldots+\posc_T(2^{\ell+1}-1) \leq (2^\ell-2)\cdot T - (2^\ell-1)\cdot m -\theta_\ell+n + \Sigma_m \\
        \implies& \posc_T(2)+\ldots+\posc_T(2^\ell) \leq (2^{\ell-1}-1)\cdot T - \frac{2^\ell-1}{2}\cdot m + \frac12\cdot (n+\Sigma_m-\theta_\ell). \numberthis \label{eqn:bound-ell-1}
    \end{align*}

    The coin in position $1$ at time $T$ can never fall off the board, because its position can only increment by $1$ in each of the last $T-n$ rounds, and so we can disregard it. We will now apply the inductive hypothesis to a new game played on the $2^\ell-1$ coins in positions $2$ through $2^\ell$ at time $T$, corresponding to the coins in positions $2$ through $2^\ell$ at time $T$. We will renumber the coins, so we can view this new game as being played on $2^\ell-1$ coins numbered $1$ through $2^\ell-1$, where the position functions are given by $\pos'$ and $\posc'$. Denote 
    \[
        \theta'_{\ell-1}:= \left\lceil \frac{1}{2^{\ell}-1} \cdot \left( (2^{\ell-1}-1)\cdot (n-T)+\posc'_T(1)+\ldots+\posc'_T(2^{\ell}-1) \right) \right\rceil.
    \]
    Our goal is to show that $\theta'_{\ell-1}\leq \theta_\ell$. This will prove the inductive step for the following reasons: 
    \begin{itemize}
        \item We have that $\posc'_T(1)=\posc_T(2)\geq \posc_T(1)=\theta_\ell-T+n\geq \theta'_{\ell-1}-T+n$, so the condition for the inductive step of the theorem is met.
    
        \item We know that $\posc_n(\ell+1) \le \max \{ \posc'_n(\ell), \theta_\ell \}$, because coin $c_T(1)$ cannot surpass $\theta_\ell$ and all of coins $c_T(2)\ldots c_T(\ell+1)$ are at most $\posc'_n(\ell)$. Moreover, by the inductive hypothesis, we know that $\posc'_n(\ell)<\theta'_{\ell-1}\leq \theta_\ell$, and so $\posc_n(\ell+1) \leq \theta_\ell$ as we wanted.
    \end{itemize}
    
    Now, we show the desired statement that $\theta'_{\ell-1}\leq \theta_\ell$. Using the bound~\eqref{eqn:bound-ell-1} we found on $\posc_T(2)+\ldots+\posc_T(2^\ell)=\posc'_T(1)+\ldots+\posc'_T(2^\ell-1)$, we have that
    \begin{align*}
        \theta'_{\ell-1}&= \left\lceil \frac{1}{2^{\ell}-1} \cdot \left( (2^{\ell-1}-1)\cdot (n-T)+\posc'_T(1)+\ldots+\posc'_T(2^{\ell}-1) \right) \right\rceil\\
        &<\left\lceil \frac{1}{2^{\ell}-1} \cdot \left( (2^{\ell-1}-1)\cdot (n-T)+(2^{\ell-1}-1)\cdot T - \frac{2^\ell-1}{2}\cdot m + \frac12\cdot (n+\Sigma_m-\theta_\ell) \right) \right\rceil \\
        &= \left\lceil \frac{1}{2^{\ell+1}-2} \cdot \left( (2^\ell-1)\cdot (n-m) +\Sigma_m-\theta_\ell \right) \right\rceil \\
        &\leq \left\lceil \frac{1}{2^{\ell+1}-2} \cdot \left( 1-\frac{1}{2^{\ell+1}-1} \right) \left( (2^\ell-1)\cdot (n-m) + \Sigma_m \right) \right\rceil \\
        &= \left\lceil \frac{1}{2^{\ell+1}-1} \cdot \left( (2^\ell-1)\cdot (n-m) + \Sigma_m \right) \right\rceil = \theta_\ell.
    \end{align*}

    This concludes the inductive step of Theorem~\ref{thm:inductive-step}.
\end{proof}

\subsection{Conclude Theorem~\ref{thm:upper-bound}}
\label{sec:conclude-upper}

Finally, let us conclude Theorem~\ref{thm:upper-bound}. Let the number of coins $K$ be at least $8\cdot(2^{\ell+1}-1)$. For Eve's first move, have her choose the one of $S_1$ and $\overline{S_1}$ (Bob's two provided sets) with fewer elements, so that at least half the coins remain at position $0$. For her second move, have her choose the one of $S_2$ and $\overline{S_2}$ that keeps the most coins at $0$, which will keep at leat $\frac14$ of the coins at $0$. She will do this one more time, so that after $3$ moves, there are at least $2^{\ell+1}-1$ coins at position $0$.

Even if there are more than $2^{\ell+1}-1$ coins at position $0$, from now, Eve will fix a subset of $2^{\ell+1}-1$ coins to play a coin game on. She wants to show that in this new $n-3$ turn coin game, she can keep $\posc_{n-3}(\ell+1) \leq \frac{2^\ell-1}{2^{\ell+1}-1}\cdot n$. This follows by setting $m=0$ in Theorem~\ref{thm:inductive-step}. The theorem states that Eve can always guarantee that $\ell+1$ coins stay at or below position $\left\lceil \frac{2^\ell-1}{2^{\ell+1}-1}\cdot (n-3) \right\rceil \leq \frac{2^\ell-1}{2^{\ell+1}-1}\cdot n$. Therefore $\left(\ell+1, \frac{2^\ell-1}{2^{\ell+1}-1}\right)$-list decoding is not possible.

\section{Optimal 2-List Feedback Codes}
\label{sec:sw-3.7}
In this section, we will prove that the feedback ECC from Spencer-Winkler~\cite{SpencerW92} achieves $\left( 2, \frac37-\epsilon \right)$-list decodability, matching the upper bound for $\ell = 2$ given in Theorem~\ref{thm:upper-bound}.

\subsection{The Spencer-Winkler Coin Game}

Let us begin by defining the Spencer-Winkler strategy for the coin game.

\protocol{The Spencer-Winkler (\SW) Coin Game}{SW}{
    Alice has a coin $x \in \{ 0, 1 \}^k$. Bob and Eve play the coin game with $K = 2^k$ coins. In each of $n$ rounds, Bob partitions the coins into two sets based on the parity of their positions on the board. That is, he numbers the coins from $1$ to $K$, where $1$ is assigned to the coin with the smallest position, and $K$ is the coin that is furthest away. The odd numbered coins are in one set, and the even numbered coins are in the other. Eve chooses a set whose coins to move up by $1$.
}

We will also denote by \SW~the feedback code corresponding to this family of coin game strategies. When defining a feedback code \SW, we must also specify a choice of the code length $n(k)$. It was shown in~\cite{SpencerW92} that \SW~achieves the unique decoding bound of $\frac13-\epsilon$ for feedback codes. 

\begin{theorem}[\cite{SpencerW92}]
    Let $n = \frac{k}{\epsilon}$. Then in the \SW~coin game, no matter which sets Eve chooses each round, it will hold that
    \[
        \posc_n(2) > \left( \frac13 - O(\epsilon) \right) \cdot n.
    \]
    Thus the \SW~coin game gives a feedback code that is uniquely decodable up to radius $\frac13 - O(\epsilon)$. 
\end{theorem}

Our goal through the rest of this section is to prove a statement about the list decodability of \SW~ to list size $2$. This proves Theorem~\ref{thm:sw-3.7}.

\begin{restatable}{theorem}{listtwo}
\label{thm:list2}
    In the \SW~coin game, no matter which sets Eve chooses each round, it holds that
    \[
        \posc_n(3) > \left( \frac37 - 2\epsilon \right) \cdot n.
    \]
    That is, the \SW~coin game with $n = \frac{24k}{\epsilon^3}$ gives a feedback code that is $\left( 2, \frac27 - 2\epsilon \right)$-list-decodable.
\end{restatable}

\subsection{Some Useful Inequalities}

Throughout this section, let $n = \frac{24k}{\epsilon^3}$. Our proof that the \SW~protocol achieves $(2, \frac37)$-list decoding follows from two inequalities, given in Lemmas~\ref{lem:x2-x1} and~\ref{lem:quadruple}. First, let us prove a statement about how the position of the $i$'th coin at each round can change.

\begin{claim} \label{claim:posc}
    For any $t \in [n-1]$ and for any $i \in [K]$, it holds that
    \[
        \posc_{t+1}(i) - \posc_t(i) \in \{ 0, 1 \}.
    \]
\end{claim}

\begin{proof}
    The $i$'th coin $c_{t+1}(i)$ at time $t+1$ being at position $\posc_{t+1}(i)$ means that at time $t$ there must've been $\ge i-1$ coins with positions $\le \posc_{t+1}(i)$ and $\ge K-i$ coins with positions $\ge \pos_t(c_{t+1}(i)) \ge \posc_{t+1}(i) - 1$. This means that the $i$'th coin at time $t$ must've had position in the range $[\posc_{t+1}(i)-1, \posc_{t+1}(i)]$, from which the claim follows.
\end{proof}

Now, we are ready to state our first key lemma, that the gaps between the positions of the first two coins must always be at least the gap between the positions of the third and fourth coins at any timestep.

\begin{lemma} \label{lem:x2-x1}
    At any point in the \SW~protocol, the following inequality holds:
    \[
        \posc_t(2) - \posc_t(1) + 1 \ge \posc_t(4) - \posc_t(3). \numberthis \label{eqn:x2-x1}
    \]
\end{lemma}

\begin{proof}
    We prove this via straightforward induction. Let us assume at time $t$ that $\posc_t(2) - \posc_t(1) + 1 \ge \posc_t(4) - \posc_t(3)$. This is certainly true at time $t = 0$, since all coins are at position $0$. 

    

    By applying Claim~\ref{claim:posc}, it follows that $\posc_{t+1}(2) - \posc_{t+1}(1) = \posc_t(2) - \posc_t(1) + \{ -1, 0, 1 \}$ and similarly $\posc_{t+1}(4) - \posc_{t+1}(3) = \posc_t(4) - \posc_t(3) + \{ -1, 0, 1 \}$. This means that $(\posc(4) - \posc(3)) - (\posc(2) - \posc(1))$ can change by at most $2$ each round, and we can focus on the relevant case that 
    \[
        (\posc_t(4) - \posc_t(3)) - (\posc_t(2) - \posc_t(t)) \in \{ 0, 1 \}. \numberthis \label{eqn:x2-x1-assume}
    \]

    We split the induction into two cases. The first is if Bob moves the odd set. In this case, Bob will move up coins $c_t(1)$ and $c_t(3)$. If $c_t(4) - c_t(3) > 0$, then this means that $c_t(4) - c_t(3)$ will decrease by $1$, and $c_t(2) - c_t(1)$ will change by at most $1$, so~\eqref{eqn:x2-x1} will still hold at time $t+1$. If $c_t(4) - c_t(3) = 0$, then $c_{t+1}(4) - c_{t+1}(3) \le 1$, and also by~\eqref{eqn:x2-x1-assume} it must be true that $\posc_t(2) - \posc_t(1) = 0$, so $c_{t+1}(2) - c_{t+1}(1) \in \{ 0, 1 \}$ as well. Therefore, $(c_{t+1}(4) - c_{t+1}(3)) - (c_{t+1}(2) - c_{t+1}(1)) \le 1$ as desired.

    The second case is if Bob moves the even set. In this case, Bob will move up coins $c_t(2)$ and $c_t(4)$. If $\posc_t(3) - \posc_t(2) > 0$, then $c_t(1), c_t(2)$, and $c_t(3)$ will still be the first, second, and third coins at time $t+1$. Then $\posc_{t+1}(2) - \posc_{t+1}(1) = (\posc_t(2) + 1) - \posc_t(1)$, while $\posc_{t+1}(4) - \posc_{t+1}(3) \le (\posc_t(4) + 1) - \posc_t(3)$ by Claim~\ref{claim:posc}, so~\eqref{eqn:x2-x1} clearly follows. On the other hand, if $\posc_t(2) = \posc_t(3)$, then moving coins the second and fourth coins at time $t$ is equivalent to moving the third and fourth coins at time $t$, so $\posc_{t+1}(2) - \posc_{t+1}(1) = \posc_t(2) - \posc_t(1)$. If $\posc_t(4) > \posc_t(3)$, then $\posc_{t+1}(3) = \posc_{t}(3) + 1$, so $\posc_{t+1}(4) - \posc_{t+1}(3) \le \posc_t(4) - \posc_t(3)$, and~\eqref{eqn:x2-x1} follows. Otherwise if $\posc_t(4) - \posc_t(3) = 0$, this means by~\eqref{eqn:x2-x1-assume} that $\posc_t(2) - \posc_t(1) = 0$, and since $\posc_{t+1}(4) = \posc_t(4) + \{ 0, 1 \}$ and $\posc_{t+1}(3) = \posc_t(3) + \{ 0, 1 \}$ by Claim~\ref{claim:posc}, we have that $\posc_{t+1}(4) - \posc_{t+1}(3) \le 1$, so again~\eqref{eqn:x2-x1} follows.
\end{proof}

The second key lemma is about the sum of the positions of the first four coins at any timestep.

\begin{lemma}\label{lem:quadruple}
    At any round $0 \le \tau \le n$, the following holds:
    \[
        \frac12 \posc_\tau(1) + \posc_\tau(2) + \posc_\tau(3) + \posc_\tau(4) \ge \frac{3\tau}2 - 5\epsilon n.
    \]
\end{lemma}

\begin{proof}
    In order to analyze the positions of four coins collectively, we define a new derivative coin game that we call the \emph{quadruple coin game}. The coins in this quadruple coin game are four-tuples of coins in the old game, for a total of $\binom{K}{4}$ coins in the new game. Let the set of coins in the original game be denoted $S$, and let the set of quadruple coins be denoted $S^\quadr$. The position of quadruple coin $(x_1, x_2, x_3, x_4) \in S^\quadr$ at time $t$ is $\frac12 \pos_t(x_1) + \pos_t(x_2) + \pos_t(x_3) + \pos_t(x_4)$, where here we assume that the four coins $x_1, x_2, x_3, x_4$ are ordered according to increasing position at time $t$. 

    We denote by $c^\quadr_t(i_1, i_2, i_3, i_4)$ the quadruple coin $(c_t(i_1), c_t(i_2), c_t(i_3), c_t(i_4))$, and by $\posc_t^\quadr(i_1, i_2, i_3, i_4) = pos_t^\quadr(c_t(i_1, i_2, i_3, i_4))$ its position at time $t$. Our goal is to show that in this quaduple coin game, coin $c_n(1,2,3,4)$ will be at position $\frac32n - 4\epsilon n$ at the end of the protocol.

    To do this, we will show that on average, the quadruple coins must be moving at $\approx \frac32$ units per round, and furthermore that no coins are left behind. What we mean by the latter is that for any quadruple coin that doesn't move enough in a given round, there is some earlier coin that is moving more than enough and compensating for the slower movement by the first coin. Precisely, we will do this via a bijection.

    Let us first define $u_t(i)$ to be the amount that coin $c_t(i)$ is moved in round $t+1$. We remark that this is \emph{not} the value $\posc_{t+1}(i)-\posc_t(i)$, since after round $t$, the coins may change their orders so that $c_t(i)$ is no longer the $i$'th coin. We will also define $u^\quadr_t(i_1, i_2, i_3, i_4)$ to be the amount that quadruple coin $c_t^\quadr(i_1, i_2, i_3, i_4)$ is moved in time $t+1$. In general, when we write a $4$-tuple, the four numbers should be in increasing order. We can also define the poset relation on $4$-tuples of indices in $[K]$ by $(i_1, i_2, i_3, i_4) < (j_1, j_2, j_3, j_4)$ if $i_\iota \le j_\iota$ for all $\iota \in [4]$.

    \begin{claim}
        At any time $t$ there is a bijection of indices $\binom{[K]}{4} \ni (i_1, i_2, i_3, i_4) \leftrightarrow (j_1, j_2, j_3, j_4) \in \binom{[K]}{4}$ such that for any bijected values $(i_1, i_2, i_3, i_4) \leftrightarrow (j_1, j_2, j_3, j_4)$, 
        \[
            u^\quadr_t(i_1, i_2, i_3, i_4) + u^\quadr_t(j_1, j_2, j_3, j_4) \ge 3,
        \]
        and furthermore if $u^\quadr_t(i_1, i_2, i_3, i_4) < \frac32$ then $(j_1, j_2, j_3, j_4) < (i_1, i_2, i_3, i_4)$.
    \end{claim}

    \begin{proof}
        We will define this bijection $\varphi_t$ by defining $\varphi_t$ on tuples $(i_1, i_2, i_3, i_4)$ for which $u^\quadr_t(i_1, i_2, i_3, i_4) < \frac32$. For all other tuples that remain unbijected after this process, we can pair them up arbitrarily.

        Consider a tuple $(i_1, i_2, i_3, i_4)$ for which $u^\quadr_t(i_1, i_2, i_3, i_4) < \frac32$. Let's suppose that Bob is moving the even coins in round $t+1$ (we'll define the bijection for the odd case in the same way, except switching the words ``even'' and ``odd'' everywhere). The only way $u^\quadr_t(i_1,i_2,i_3,i_4)$ can be less than $\frac32$ is if there is at most one even index among the four. Let us define the map $\varphi_t$ depending on which of the $i_\iota$ is even.
        \[
        \varphi_t(i_1, i_2, i_3, i_4) = 
        \begin{cases}
            (i_1, i_2-1, i_3-1, i_4-1) & \text{$i_1, i_2, i_3, i_4$ all odd} \\
            (i_1, i_2, i_3-1, i_4) & \text{$i_1,i_2,i_3$ odd and $i_4$ even} \\
            (i_1, i_2-1, i_3, i_4) & \text{$i_1,i_2,i_4$ odd and $i_3$ even} \\
            (i_1, i_2, i_3, i_4-1) & \text{$i_1,i_3,i_4$ odd and $i_2$ even} \\
            (i_1, i_2, i_3, i_4-1) & \text{$i_2,i_3,i_4$ odd and $i_1$ even}.
        \end{cases}
        \]
        It's straightforward to check that these are all well defined, and that each of the bijected tuples contains four distinct elements. (Recall we always assume $i_1 < i_2 < i_3 < i_4$.)

        Next, we also need to check that $\varphi_t(i_1, i_2, i_3, i_4) \not= \varphi_t(i'_1, i'_2, i'_3, i'_4)$ for all $(i_1, i_2, i_3, i_4) \not= (i'_1, i'_2, i'_3, i'_4)$. To see this, we show how to invert each tuple in the image of the map $\varphi_t$ as we've defined so far. The key observation is that the five cases give five different parity tuples for the output tuple: for instance, the first case that $i_1, i_2, i_3, i_4$ are all odd results in an output where the parities are (odd, even, even, even), whereas all the rest of the maps give outputs where two indices are odd and two are even, so we always know how to invert a tuple that is (odd, even, even, even).  
        \[
        \varphi_t^{-1}(j_1, j_2, j_3, j_4) = 
        \begin{cases}
            (j_1, j_2+1, j_3+1, j_4+1) & \text{$j_1$ odd and $j_2, j_3, j_4$ odd} \\
            (j_1, j_2, j_3+1, j_4) & \text{$j_1,j_2$ odd and $j_3,j_4$ even} \\
            (j_1, j_2+1, j_3, j_4) & \text{$j_1,j_4$ odd and $j_2,j_3$ even} \\
            (j_1, j_2, j_3, j_4+1) & \text{$j_1,j_3$ odd and $j_2,j_4$ even} \\
            (j_1, j_2, j_3, j_4+1) & \text{$j_2,j_3$ odd and $j_1,j_4$ even}.
        \end{cases}
        \]

        Finally, we check that $u^\quadr_t(i_1, i_2, i_3, i_4) + u^\quadr_t(\varphi_t(i_1, i_2, i_3, i_4)) \ge 3$. In the first case that $i_1,i_2,i_3,i_4$ are all odd and $^\quadr_t(i_1, i_2, i_3, i_4) = 0$, $\varphi_t(i_1,i_2,i_3,i_4)$ has three even indices, so $u^\quadr_t(\varphi_t(i_1, i_2, i_3, i_4)) = 3$. In all the other cases, $(i_1, i_2, i_3, i_4)$ has one even index and $\varphi_t(i_1, i_2, i_3, i_4)$ has two even indices so $u^\quadr_t(i_1, i_2, i_3, i_4) + u^\quadr_t(\varphi_t(i_1, i_2, i_3, i_4)) = 1 + 2 = 3$.
    \end{proof}

    Now, let us define the potential function 
    \[
        \psi(t) = \sum_{(x_1, x_2, x_3, x_4) \in \binom{\{ 0, 1 \}^k}{4}} w_t(x_1, x_2, x_3, x_4),
    \]
    where 
    \[ 
        w_t(x_1, x_2, x_3, x_4) = (1 + \epsilon)^{\frac13 \cdot d_t(x_1, x_2, x_3, x_4)}
    \]
    and 
    \[
        d_t(x_1, x_2, x_3, x_4) = \max \left\{ 0, \frac{3t}{2} - \pos^\quadr_t(x_1, x_2, x_3, x_4) \right\}.
    \]
    Since $\pos^\quadr_0(x_1, x_2, x_3, x_4) = 0$ for all $(x_1, x_2, x_3, x_4)$ at the start of the game, $\psi(0) = \binom{K}{4}$.

    \begin{claim} \label{claim:sum-bijected}
        For any time $t$, let us consider any bijected pair $(i_1, i_2, i_3, i_4) \leftrightarrow (j_1, j_2, j_3, j_4)$ under $\varphi_t$. Define $(x_1, x_2, x_3, x_4) = c^\quadr_t(i_1, i_2, i_3, i_4)$ and $(y_1, y_2, y_3, y_4) = c^\quadr_t(j_1, j_2, j_3, j_4)$. Then,
        \[
            w_{t+1}(x_1, x_2, x_3, x_4) + w_{t+1}(y_1, y_2, y_3, y_4)
            \le e^{\epsilon^2/2} \cdot \left( w_t(x_1, x_2, x_3, x_4) + w_t(y_1, y_2, y_3, y_4) \right) + \epsilon,
        \]
    \end{claim}

    \begin{proof}
        We proceed by casework on whether $u_t(x_1, x_2, x_3, x_4)$ and $u_t(y_1, y_2, y_3, y_4)$ are greater than or less than $\frac32$.

        \begin{caseof}
        \case{Both $u_t(x_1, x_2, x_3, x_4)$ and $u_t(y_1, y_2, y_3, y_4)$ are at least $\frac32$.}{
            In this case, $d_{t+1}(x_1, x_2, x_3, x_4) \le d_t(x_1, x_2, x_3, x_4)$ and $d_{t+1}(y_1, y_2, y_3, y_4) \le d_t(y_1, y_2, y_3, y_4)$, so
            \begin{align*}
                w_{t+1}(x_1, x_2, x_3, x_4) + w_{t+1}(y_1, y_2, y_3, y_4) 
                &\le w_t(x_1, x_2, x_3, x_4) + w_t(y_1, y_2, y_3, y_4) \\
                &\le e^{\epsilon^2/2} \cdot \left( w_t(x_1, x_2, x_3, x_4) + w_t(y_1, y_2, y_3, y_4) \right) + \epsilon.
            \end{align*}
        }
        \case{$u_t(x_1, x_2, x_3, x_4) < \frac32$.}{
            We write $u_t(x_1, x_2, x_3, x_4) = \frac32 - \delta$ for some $\delta > 0$. Then $d_{t+1}(x_1, x_2, x_3, x_4) \le d_t(x_1, x_2, x_3, x_4) + \delta$ and $u_t(y_1, y_2, y_3, y_4) \ge \frac32 - \delta$. Let us again do casework on whether $d_t(y_1, y_2, y_3, y_4)$ is larger or smaller than $\delta$.
            \subcase{$d_t(y_1, y_2, y_3, y_4) \ge \delta$.}{
                In this case, $d_{t+1}(y_1, y_2, y_3, y_4) \le d_t(y_1, y_2, y_3, y_4) - \delta$. Then
                \begin{align*}
                    w_{t+1}(x_1, x_2, x_3, x_4) + w_{t+1}(y_1, y_2, y_3, y_4) 
                    &\le (1+\epsilon)^{\frac13 \delta} w_t(x_1, x_2, x_3, x_4) + (1+\epsilon)^{-\frac13 \delta} w_t(y_1, y_2, y_3, y_4) \\
                    &\le \frac{(1 + \epsilon)^{\frac13 \delta} + (1 + \epsilon)^{- \frac13 \delta}}{2} \cdot \left( w_t(x_1, x_2, x_3, x_4) + w_t(y_1, y_2, y_3, y_4) \right) \\
                    &\le \frac{(1 + \epsilon) + (1 + \epsilon)^{-1}}{2} \cdot \left( w_t(x_1, x_2, x_3, x_4) + w_t(y_1, y_2, y_3, y_4) \right) \\
                    &\le e^{\epsilon^2/2} \left( w_t(x_1, x_2, x_3, x_4) + w_t(y_1, y_2, y_3, y_4) \right) \\
                    &\le e^{\epsilon^2/2} \left( w_t(x_1, x_2, x_3, x_4) + w_t(y_1, y_2, y_3, y_4) \right) + \epsilon,
                \end{align*}
                where the second inequality follows from $(j_1, j_2, j_3, j_4) < (i_1, i_2, i_3, i_4) \implies w_t(x_1, x_2, x_3, x_4) \le w_t(y_1, y_2, y_3, y_4)$ and the rearrangement inequality, and the thir inequality follows from $\frac13 \delta \le \frac12 < 1$. 
            }
            \subcase{$d_t(y_1, y_2, y_3, y_4) < \delta$.}{
                In this case, $d_{t+1}(y_1, y_2, y_3, y_4) = 0$, and $d_{t+1}(x_1, x_2, x_3, x_4) \le d_t(x_1, x_2, x_3, x_4) + \delta \le d_t(y_1, y_2, y_3, y_4) + \delta < 2\delta$. Then
                \begin{align*}
                    w_{t+1}(x_1, x_2, x_3, x_4) + w_{t+1}(y_1, y_2, y_3, y_4) 
                    &\le (1+\epsilon)^{\frac23 \cdot 2\delta} + (1 + \epsilon)^0 \\
                    &\le 2 + \epsilon \\
                    &\le \left( w_t(x_1, x_2, x_3, x_4) + w_t(y_1, y_2, y_3, y_4) \right) + \epsilon \\
                    &\le e^{\epsilon^2/2} \left( w_t(x_1, x_2, x_3, x_4) + w_t(y_1, y_2, y_3, y_4) \right) + \epsilon.
                \end{align*}
            }
        }
        \end{caseof}
    \end{proof}

    Now, using Claim~\ref{claim:sum-bijected} and summing over all bijected tuples $(i_1, i_2, i_3, i_4) \leftrightarrow (j_1, j_2, j_3, j_4)$, we get 
    \[
        \psi(t+1) \le e^{\epsilon^2/2} \cdot \psi(t) + \epsilon \cdot \frac{\binom{K}{4}}{2}.
    \]
    Solving the reccurence, we get 
    \begin{align*}
        \psi(t+1) 
        &\le e^{(t+1)\epsilon^2/2} \cdot \binom{K}{4} + \frac{e^{(t+1)\epsilon^2/2}-1}{e^{\epsilon^2/2} - 1} \cdot \epsilon \cdot \frac{\binom{K}{4}}{2} \\
        &\le e^{(t+1)\epsilon^2/2} \cdot \binom{K}{4} \cdot \left( 1 + \frac{1}{2(e^{\epsilon^2/2}-1)} \right) \\
        &\le e^{(t+1)\epsilon^2/2} \cdot \binom{K}{4} \cdot \left( 1 + \frac1\epsilon \right).
    \end{align*}
    
    Finally, we get our bound on $\frac12 \posc_\tau(1) + \posc_\tau(2) + \posc_\tau(3) + \posc_\tau(4)$. Let $x_1 = c^\quadr_\tau(1)$, $x_2 = c^\quadr_\tau(2)$, $x_3 = c^\quadr_\tau(3)$, and $x_4 = c^\quadr_\tau(4)$. We have that
    \begin{align*}
        w_\tau(x_1, x_2, x_3, x_4) = (1 + \epsilon )^{\frac13 \cdot d_\tau(x_1, x_2, x_3, x_4)}
        &< \psi(\tau) \\
        &\le e^{\tau \epsilon^2/2} \cdot \binom{K}{4} \cdot \left( 1 + \frac1\epsilon \right) \\
        &\le e^{\tau \epsilon^2/2} \cdot K^4 \cdot \left( 1 + \frac1\epsilon \right),
    \end{align*}
    which gives 
    \begin{align*}
        \frac{3\tau}{2} - \pos^\quadr_\tau(x_1, x_2, x_3, x_4) 
        &\le d_\tau(x_1, x_2, x_3, x_4) \\
        &< \frac{3}{\ln (1 + \epsilon)} \cdot \left( \tau\epsilon^2/2 + 4 \ln K + \ln \left( 1 + \frac1\epsilon \right) \right) \\
        &\le \frac{6}{\epsilon} \cdot \left( n \epsilon^2/2 + 4 \ln K + \frac1\epsilon \right) \\
        &\le 3\epsilon n + \frac{6}{\epsilon} \cdot \left( 4 \ln K + \frac1\epsilon \right) \\
        &\le 5\epsilon n,
    \end{align*}
    where we use that $\tau \le n$, $\frac{\epsilon}2 \le \ln(1 + \epsilon)$, and $\ln \left( 1 + \frac1\epsilon \right) \le 1/\epsilon$, and where the last inequality follows from $n = \frac{24k}{\epsilon^3}$ so that
    \begin{align*}
        \frac{6}{\epsilon} \cdot \left( 4 \ln K + \frac1\epsilon \right) 
        \le \frac{6}{\epsilon} \left( 4 k + \frac1\epsilon \right) 
        = \epsilon^2 n + \frac\epsilon4 n 
        < 2\epsilon n.
    \end{align*}

    This completes the proof of Lemma~\ref{lem:quadruple}.
\end{proof}

\subsection{The $(2, \frac37-\epsilon)$-List Decodability of \SW}

We now return to our main goal of this section, which is to show that \SW~is $2$-list decodable up to radius $\frac37$. We restate the theorem below for convenience.

\listtwo*

\begin{proof}
    Consider the first round $n-f$ at which $\posc_{n-f}(4) > \left( \frac37 - 2\epsilon \right) n$. There must be such a round because $\frac12 \posc_n(1) + \posc_n(2) + \posc_n(3) + \posc_n(4) \ge \frac32 n - 5\epsilon n$ by Lemma~\ref{lem:quadruple}, so $\posc_n(4) \ge \frac27 \cdot \left( \frac32 n - 7\epsilon n \right) = \left( \frac37 - 2\epsilon \right) n$. After round $n-f$, there only remains $3$ relevant coins upon the board. Our goal is to show that another falls off.

    Suppose that in round $n - f$ that the first coin is at position $\left( \frac37 - 2\epsilon \right) n - p$. From Lemma~\ref{lem:x2-x1}, we know that 
    \begin{align*}
        \posc_{n-f}(2) + \posc_{n-f}(3) 
        &\ge \posc_{n-f}(1) + \posc_{n-f}(4) - 1 \\
        &= \frac67 n - p - 4\epsilon n - 1 \\
        &\ge \frac67n - p - 5\epsilon n. \numberthis \label{eqn:x2+x3-(1)}
    \end{align*}
    We also know from Lemma~\ref{lem:quadruple} that 
    \begin{align*}
        \posc_{n-f}(2) + \posc_{n-f}(3) 
        &\ge \frac32 (n-f) - 5\epsilon n - \frac12 \posc_{n-f}(1) - \posc_{n-f}(4) \\
        &\ge \frac32 (n-f) - 5\epsilon n - \frac12 \left( \left( \frac37 - 2\epsilon \right) n - p \right) - \left( \frac37 - 2\epsilon \right) n \\
        &= \frac67 n - \frac32 f + \frac12 p - 2\epsilon n. \numberthis \label{eqn:x2+x3-(2)}
    \end{align*}
    Adding $\frac13$ of Equation~\ref{eqn:x2+x3-(1)} and $\frac23$ of Equation~\ref{eqn:x2+x3-(2)}, we get
    \[
        \posc_{n-f}(2) + \posc_{n-f}(3) \ge \frac67 n - f - 3\epsilon n. \numberthis \label{eqn:x2+x3-(3)}
    \]

    Now, suppose that the third coin does not reach position $\left( \frac37 - 3\epsilon \right) n$ by the end of the game. Then, at least one of the second and third coins must move in each of the last $f$ rounds, so that $\posc_t(2) + \posc_t(3)$ increases by at least $1$ each round. This means that at round $n$,
    \begin{align*}
        \posc_n(2) + \posc_n(3) 
        &\ge \posc_{n-f}(2) + \posc_{n-f}(3) + f\\
        &\ge \left( \frac67 - 3\epsilon \right) n,
    \end{align*}
    which gives that $\posc_n(3) \ge \frac12 \left( \frac67 - 3\epsilon \right) n > \left( \frac37 - 2\epsilon \right) n$, a contradiction on our assumption that the third coin never falls off the board.

    Therefore, no matter which sets the adversary chooses, there will always be at most two coins left on the board at the end of the game.
    
\end{proof}

\section{Spencer-Winkler Does Not Give $\left( 3, \frac{7}{15} \right)$-List Feedback Codes}
In Section~\ref{sec:sw-3.7}, we showed that the Spencer-Winkler strategy for Alice and Bob achieves a $(2,\frac37)$-list feedback code. In this strategy, Bob always splits the sets into the even-numbered coins and the odd-numbered coins. 

A reasonable conjecture might be that this strategy is always optimal, and, in fact, we do not know of any proof that it is not. However, in this section we will show that for $\ell=3$, it does not match our upper bound of $\frac7{15}$ from Section~\ref{sec:upper}. Whether this is because the protocol is suboptimal or because the upper bound is not tight remains open for further investigation.

\counterexample*

As in the previous sections, we discuss the counterexample in the coin game formulation. We will show an attack for all large enough $k$. In this setting, Bob has $K=2^k$ coins for some $k>1000\eps^{-1}$, and in each of $n$ rounds (Alice and Bob choose $n$), he partitions the coins into $2$ sets. Then, Eve chooses one of the two sets in which to move all the coins up by $1$. Eve wins if, at the end of the protocol, at least $4$ coins remain $\lesssim \frac{31}{67}n$. 

\subsection{Eve's Strategy}

We begin by stating Eve's procedure for determining which set of coins to select at each step. For simplicity, we'll assume that $n$ is a multiple of $67$ and write $n=67q$. We'll further assume $q$ is even and $\eps q$ is an integer.

\protocol{Eve's Procedure}{SW}{
    Alice has a coin $x \in \{ 0, 1 \}^k$. She and Bob simulate the coin game with $K = 2^k$ coins. For the $n$ rounds, Eve does the following:

    \begin{enumerate}
        \item For the first $32q$ rounds, Eve selects (moves up by $1$) the even coins. 
        \item For the next $16q$ rounds, Eve selects the odd coins. 
        \item For the next $8q$ rounds, Eve selects the even coins. 
        \item For the next $4q$ rounds, Eve selects the odd coins. 
        \item For the next $4q$ rounds, Eve selects the even coins. 
        \item For the final $3q$ rounds, Eve selects the odd coins. 
    \end{enumerate}
}

\subsection{Analysis of Eve's Strategy}

We will establish our main sequence of claims giving an upper bound on the positions of the first few coins after each of these steps (and a step in between as well). The claims are written out below and summarized by the following table. For each upper bound, we additionally add $+\eps q$, although we do not write it into the table for brevity.

\begin{figure}[ht]
\centering

\renewcommand{\arraystretch}{1.5}
\begin{tabular}{c!{\vrule width 1.5pt}c|c|c|c|c} 
  & $\posc_m(1)$ & $\posc_m(2)$ & $\posc_m(3)$ & $\posc_m(4)$ & $\posc_m(5)$ \\
  \specialrule{1.5pt}{0pt}{0pt} 
  (even) $m=32q$ & 0 ($\blacktriangle$) & 16q ($\bullet$) & 16q ($\bullet$) & 16q ($\bullet$) & 16q ($\bullet$) \\
  \hline
  (odd) $m=48q$ & 16q ($\rule{.5em}{.5em}$) & 16q ($\blacklozenge$) & 24q ($\bullet$) & 24q ($\bullet$) & 24q ($\bullet$) \\
  \hline
  (even) $m=56q$ & 16q ($\blacktriangle$) & 24q ($\rule{.5em}{.5em}$) & 24q ($\blacklozenge$) & 28q ($\bullet$) & 28q ($\bullet$)\\
  \hline
  (odd) $m=60q$ & 20q ($\rule{.5em}{.5em}$) & 24q ($\blacklozenge$) & 28q ($\rule{.5em}{.5em}$) & 28q ($\blacklozenge$) & 30q ($\bullet$)\\
  \hline
  (even) $m=62q$ & 20q ($\blacktriangle$) & 26q ($\rule{.5em}{.5em}$) & 28q ($\blacklozenge$) & 30q ($\rule{.5em}{.5em}$) & 30q ($\blacklozenge$) \\
  \hline
  (even) $m=64q$ & 20q ($\blacktriangle$) & 28q ($\rule{.5em}{.5em}$) & 28q ($\blacklozenge$) & 31q ($\clubsuit$) & 31q ($\clubsuit$) \\
  \hline
  (odd) $m=67q$ & 23q ($\rule{.5em}{.5em}$) & 28q ($\blacklozenge$) & 31q ($\rule{.5em}{.5em}$) & 31q ($\blacklozenge$) & 33.5q ($\bullet$) \\
\end{tabular}

\caption{An upper bound on the positions of the first 5 coins after $m$ rounds. To each listed position bound, additionally add $+\eps q$. Eve's selection of odd or even is also noted. Symbols in the cells are referenced in the proofs below.} 
\end{figure}

Establishing the bounds in the table would show that $\posc_m(4)\leq 31q+\eps q < \left( \frac{31}{67} +\eps \right) \cdot n$ at the end of the $n$ rounds, which establishes Theorem~\ref{thm:counterexample}.

Now, we will prove the bounds listed in the table. Each argument below will establish the bounds indicated by the corresponding symbol. Each argument requires only that the bounds in the previous line of the table hold, so we can prove the bounds in a top-down order.

To show the ($\blacktriangle$) bounds, note that selecting the even coins never moves the coin in position $1$.

Next, we establish the ($\bullet$) bounds with the following claim. It shows us that for all times $m$, for $i\leq 6$ we have that $\posc_m(i)<\frac12 \cdot m + \eps q$

\begin{claim} \label{lem:12-bound}
    For any round $m$, as long as $k$ is sufficiently large, it holds that $\posc_m(6)<\frac12 \cdot m + \eps q$.
\end{claim}

\begin{proof}
    Since at all steps, only the odd coins or even coins move (and there are an even number of coins total), it holds that 
    \[
        \posc_m(1)+\ldots+\posc_m(K) \leq \frac12 \cdot Km
    \]
    Then, 
    \[
        \frac12 \cdot Km > \posc_m(6)+\ldots+\posc_m(K)\geq (K-6)\cdot \posc_m(6)
    \]
    which implies that
    \[
        \posc_m(6) < \frac12 \cdot \frac{Km}{K-6} < \frac12 \cdot m + \eps q
    \]
    where the last line holds if $K=2^k$ is sufficiently large, for example larger than $1000\eps^{-1}$.
\end{proof}

To show the ($\rule{.5em}{.5em}$) bounds, note the following claim.

\begin{claim} \label{clm:1perstep}
    For all $i\in [K]$, between rounds $t$ and $t'$, it holds that $\posc_{t'}(i)-\posc_t(i)\leq t'-t$.
\end{claim}

\begin{proof}
    All of coins $c_t(1),\ldots,c_t(i)$ must have have position at most $\posc_t(i)$ after round $t$. They can all move at most $t'-t$ by round $t'$, so there are at least $t'$ coins at position $\posc_t(i)+(t'-t)$ or earlier by the end of round $t'$, which is what we desired.
\end{proof}

Next, we establish the ($\blacklozenge$) bounds with the following claim.

\begin{claim} \label{clm:nocatchup}
    For all coins $i\in [K-1]$ and times $t$ and $t'$, if the set corresponding to parity $i+1$ is not chosen in that interval, then $\posc_{t'}(i+1)\leq \max \left\{ \posc_t(i+1), \posc_t(i)+t'-t \right\}$.
\end{claim}

\begin{proof}
    Let $d:= \posc_t(i+1) - \posc_t(i)$. For the first $d$ steps then, the coins $c_t(1),\ldots,c_t(i)$ are all necessarily before the coin $c_t(i+1)$. Then, $\posc_t(i+1)$ does not increase for those $d$ rounds (if $t'-t<d$, then simply stop at round $t'$. Then, by Claim~\ref{clm:1perstep}, in the remaining $\max \{t'-t-d, 0 \}$ rounds, $\posc_t(i+1)$ increases by at most $1$. Therefore, in total, the final position of $c_t(i+1)$ after round $t'$ is $\max \{t'-t-d + \posc_t(i+1), \posc_t(i+1) \}$. is This gives the desired bound because $t'-t-d +\posc_t(i+1) \leq t' -t + \posc_t(i) - \posc_t(i+1) + \posc_t(i+1) = t' -t + \posc_t(i)$ as desired.
\end{proof}

Finally, we establish the ($\clubsuit$) bounds specifically for $\posc_m(4)$ and $\posc_m(5)$ between rounds $62q$ and $64q$.

\begin{claim}
    After round $m+d$ where $m=62q$ and $d$ is even and $d\leq 2q$, it holds that $\posc_{m+d}(5)\leq 30q+d/2 + \eps q$.
\end{claim}

\begin{proof}
We will show this statement inductively. This holds when $d=0$ as the base case.

Assume this is true for $d$, and we'll prove it for $d+2$. By Claim~\ref{clm:1perstep}, $\posc(5)$ can increase by at most $2$ in these $2$ steps, so we only have to address the case where $\posc_{m+d}(5)=30q+d/2 +\eps q$ exactly. Moreover, we must have $\posc_{m+d}(4)=30q+d/2 +\eps q$ or $\posc_{m+d}(4)=30q+d/2 + \eps q$ exactly, otherwise we would be done by Claim~\ref{clm:nocatchup}.

All the coins $c_m(1),c_m(2),c_m(3)$ are in position at most $28q+d+\eps q<30q+d/2+\eps q-10$ at the end of $d$ steps, so they cannot surpass $c_{m+d}(4)$ during these next $2$ steps. 

In both the cases where $\posc_{m+d}(4)=30q+d/2 + \eps q$ and $\posc_{m+d}(4)=30q+d/2 + \eps q$, both the 4th and 5th positions move up by exactly $1$ if the first three coins are guaranteed to be out of the picture.
\end{proof}

\bibliographystyle{alpha}
\bibliography{refs}

\end{document}